\documentclass[journal,10pt]{IEEEtran}

\usepackage{cite}
\usepackage[cmex10]{amsmath}
\usepackage{amssymb}
\usepackage{amsthm}
\usepackage{array}

\newcommand{\FF}{{\mathbb F}}
\newtheorem{lemma}{Lemma}
\newtheorem{theorem}{Theorem}
\newtheorem{corollary}{Corollary}
\theoremstyle{definition}
\newtheorem{remark}{Remark}

\begin{document}
%
\title{The existence of perfect codes in Doob graphs}
%
%

\author{Denis~S.~Krotov
\thanks{D. S. Krotov is with the Sobolev Institute of Mathematics, Novosibirsk 630090 Russia e-mail: {\tt krotov@math.nsc.ru}}
\thanks{This work was funded by the Russian Science Foundation (18-11-00136).}%
\thanks{The results of this work were presented in part at the Sixteenth International Workshop on Algebraic and Combinatorial Coding Theory, Svetlogorsk, Russia, 2--8 September 2018.}%
\thanks{This is the accepted version of the paper published in the IEEE Transactions on Information Theory,  66:3 (2020), 1423--1427, \texttt{https://doi.org/10.1109/TIT.2019.2946612}  \copyright~2019 IEEE}%
}

\markboth{The existence of perfect codes in Doob graphs}%
{The existence of perfect codes in Doob graphs}

\maketitle

\begin{abstract} \boldmath
We solve the problem of existence of perfect codes in the Doob graph. It is shown that $1$-perfect codes in the Doob graph $D(m,n)$ exist if and only if $6m+3n+1$ is a power of $2$; that is, if the size of a $1$-ball divides the number of vertices. 
\end{abstract}


\section{Introduction}
{T}{he} codes in Doob graphs are special cases of codes over Eisenstein--Jacobi integers, see, e.g., \cite{Huber94,MSBG:2008}, which can be used for the information transmission in the channels with two-dimensional or complex-valued modulation. The vertices of a Doob graph can be considered as words in the mixed alphabet consisting of elements of the quotient (modulo $4$ and modulo $2$) rings of the ring of Eisenstein--Jacobi integers, see, e.g., \cite{Kro:perfect-doob}. In contrast to the cases considered in \cite{Huber94,MSBG:2008}, the number $4$ is not prime, and the quotient ring is not a field. This fact is not a problem from the point of view of the modern coding theory, which has a rich set of algebraic and combinatorial tools to deal with rings, see, e.g., \cite{SAS:codes&rings}; moreover, studying codes in Doob graphs is additionally motivated by the application of association schemes in coding theory \cite{Delsarte:1973}: 
the algebraic parameters of the schemes associated with these graphs are the same as for the quaternary Hamming scheme (this fact can be also treated from the point of view of the corresponding distance-regular graphs).

In this paper, we completely solve the problem of existence of perfect codes in the class of Doob graphs.
Namely,
we show the existence of $1$-perfect codes
in the Doob graph $D(m,n)$ for all $m$ and $n$ that satisfy the obvious necessary condition: the size $6m+3n+1$ of a ball of radius $1$ divides the number $4^{2m+n}$ of vertices. In the previous papers \cite{KoolMun:2000,Kro:perfect-doob,SHK:additive}, the problem was solved only for the cases when the parameters satisfy additional conditions admitting the existence of linear or additive perfect codes, or for small values of $m$.

The class of Doob graphs is a class of  distance-regular graphs of unbounded diameter, and the problem considered can be viewed in the general context of the problem of existence of perfect codes in distance-regular graphs. We mention some known results in this area,
mainly concentrating on distance-regular graphs important for coding theory. 
A connected graph is called distance-regular if there are constants $s_{ij}$ such that for every $i$, $j$ and for every vertex $x$, every vertex $y$ at distance $i$ from $x$ has exactly $s_{ij}$ neighbors at distance $j$ from $x$.
In the Hamming graphs $H(n,q)$, the problem of complete characterization of parameters of perfect codes is solved only for the case when $q$ is a prime power \cite{ZL:1973,Tiet:1973}: there are no nontrivial perfect codes except the $e$-perfect repetition codes in $H(2e+1,2)$, the $3$- and $2$-perfect Golay codes \cite{Golay:49} in $H(23,2)$ and $H(11,3)$, respectively, and the $1$-perfect codes in $H((q^k-1)/(q-1),q)$. In the case of a non-prime-power $q$, no nontrivial perfect codes are known, and the parameters for which
the nonexistence is not proven are restricted by $1$- and $2$-perfect codes (the last case is solved for some values of $q$), see \cite{Heden:2010:non-prime} for a survey of the known results in this
area. The problem of the (non)existence of perfect codes in the Johnson graphs $J(n,w)$ is known as  Delsarte's conjecture, see \cite{Etzion:2007:Johnson} and \cite{Gordon:2006} for the known nonexistence results; in general, the problem remains open. 
An interesting open problem is connected with the problem of existence of $1$-perfect codes in the doubled Johnson (doubled Odd) graph 
$J(2w+1,w,w+1)$ (the subgraph of $H(2w+1,2)$ induced by the words of weight $w$ and $w+1$): the existence of such codes is equivalent to the existence of Steiner systems $S(w,w+1,2w+2)$; in particular, the Steiner quadruple system $S(3,4,8)$ and the small Witt design $S(5,6,12)$
\cite{Carmichael:31,Witt:37} correspond to nontrivial $1$-perfect codes in $J(7,3,4)$ and $J(11,5,6)$;
 the nonexistence of Steiner systems 
$S(4,5,15)$ \cite{MendHung72} and $S(4,5,17)$ \cite{OstPot:S-4-5-17} 
implies the nonexistence of $1$-perfect codes in 
$J(19,9,10)$ and $J(23,11,12)$ 
(in general, the problem remains open, with the first open case in $J(31,15,16)$).
In the Grassmann graphs $J_q(n,w)$ and the bilinear forms graphs $B_q(m,n)$, nontrivial perfect codes do not exist \cite{Chihara87}, see also~\cite{MarZhu:1995}.
Some perfect codes in dual polar graphs 
are discussed in \cite[p.659]{Stanton1980},
including the examples of $1$-perfect codes 
found in \cite{Thas1977} 
in graphs of diameter $3$.  
Studying diameter-$3$ antipodal distance-regular graphs with $1$-perfect codes 
(usually, with some assumptions on the graph automorphisms)
is a separate topic, 
see \cite{GLP:1998,MakPadTsi:2013,MakPadTsi:2014,Tsiovkina:2015,MakPadTsi:2018,Tsiovkina:2015b,Tsiovkina:2017}.

The Doob graph $D(m, n)$ is the Cartesian product 
of $m$ copies of the Shrik\-hande graph
and $n$ copies of the complete graph of order $4$ 
(detailed definitions are given in the next
section). 
It is a distance-regular graph of diameter $2m + n$ 
with the same parameters (intersection array) as the
Hamming graph $H(2m + n, 4)$.
On the other hand, 
the vertices of the Doob graph can be naturally associated 
with the elements of the module 
$\mathrm{GR}(4^2)^m \times \FF_4^n$ over the Galois ring $\mathrm{GR}(4^2)$
or with the elements of the module 
${\mathbb Z}_4^{2m} \times {\mathbb Z}_2^{2n'} \times {\mathbb Z}_4^{n''}$ over ${\mathbb Z}_4$, where $n'+n''=n$.
In this way, 
the Doob graph is a Cayley graph on the corresponding module.
The submodules of the first module are called the linear codes in $D(m,n)$;
the submodules of ${\mathbb Z}_4^{2m} \times {\mathbb Z}_2^{2n'} \times {\mathbb Z}_4^{n''}$ 
are called the additive codes in $D(m,n)$.
The history of studying perfect codes in Doob graphs started from the paper
\cite{KoolMun:2000}, where it was shown that nontrivial $e$-perfect codes 
in $D(m, n)$ can only exist when $e = 1$ and $2m + n = (4^k- 1)/3$ 
for some integer $k$ and two $1$-perfect codes, in $D(2,1)$ and $D(1,3)$, 
were constructed.
In~\cite{Kro:perfect-doob}, infinite series of perfect codes in Doob graphs were obtained. In particular, it was shown that the necessary condition $2m + n = (4^k-1)/3$ is sufficient if $m < n-o(2m+n)$;
the class of linear perfect codes was completely characterized; a class of additive perfect codes was constructed and necessary conditions on $m$, $n'$, $n''$ for the existence of additive perfect codes in $D(m,n'+n'')$ were obtained (in a recent work \cite{SHK:additive}, it was shown that those conditions are also sufficient). 

\section{Definitions}

The Shrikhande graph $\mathrm{Sh}$ can be naturally defined on the pairs of elements from ${\mathbb Z}_4$. Two such pairs $(x_1,x_2)$ and $(y_1,y_2)$ are adjacent if their difference 
$(x_1-y_1,x_2-y_2)$ is one of $(0,1)$, $(0,3)$, $(1,0)$, $(3,0)$, $(1,1)$, $(3,3)$ (so, $\mathrm{Sh}$ is a Cayley graph on ${\mathbb Z}_4^2$).

We will use two representations of the complete graph $K_4$.
In the first one, $K_4({\mathbb Z}_4)$, 
its vertices are the elements $0$, $1$, $2$, $3$ of ${\mathbb Z}_4$; 
in the second, $K_4(\FF_4)$, the elements $0$, $1$, $\xi$, $\xi^2$ 
of the finite field $\FF_4$ of order $4$.

If $m$ is even, then $D(m,n)$ will be considered 
as the Cartesian product of $m$ copies of $\mathrm{Sh}$ 
and $n$ copies of $K_4(\FF_4)$ 
(in particular, $D(0,n)$ is the Hamming graph $H(n,4)$).
If $m$ is odd, then  $D(m,n)$ will be considered 
as the Cartesian product 
of $m$ copies of $\mathrm{Sh}$, 
two copies of $K_4({\mathbb Z}_4)$ 
and $n-2$ copies of $K_4(\FF_4)$
(note that in the considered class of parameters, 
$6m+3n+1$ is a power of $4$,
so $n$ is odd and $n=1$ implies even $m$).
So, the vertex set is the set of words of length $2m+n$ 
from $({\mathbb Z}_4^{2})^m\times \FF_4^{n}$ or
$({\mathbb Z}_4^{2})^m\times {\mathbb Z}_4^2\times \FF_4^{n-2}$,
and two vertices are adjacent if their coordinatewise difference has exactly one non-zero position $i$, $i>2m$,  or exactly one non-zero position $i$, $i\le 2m$, with value $1$ or $3$,
or exactly two nonzero positions
$2i-1$, $2i$, where $i\in \{1,\ldots,m\}$,
with values $1,1$ or $3,3$.

The distance between 
two vertices $\bar x$ and $\bar y$ of $D(m,n)$ (as well as in any other connected graph) 
is defined as the number of edges in the shortest path connecting $\bar x$ and $\bar y$.
Equivalently, the distance is equal to the sum of distances between 
the corresponding components of $\bar x$ and $\bar y$: 
$m$ Shrikhande components and $n$ $K_4$-components.
The distance form some vertex $\bar x$ of $D(m,n)$ 
to the all-zero vertex of $D(m,n)$ is referred to as the weight
of $\bar x$.

In any graph, an $e$-perfect code is defined as a set of vertices such that every ball of radius $e$ contains exactly one code vertex. We define a $1$-perfect Hamming code $\mathcal H$ in $H(n,4)$, $n=(4^k-1)/3$, by the check matrix consisting of all columns of height $k$ whose first nonzero element is $1$. To be explicit, we require the columns to be inverse-lexicographically ordered, for example ($k=3$),
$$\!\!\!\left[\!\!
\begin{array}{c@{\ }c@{\ }c@{\ }c@{\ }c@{\ }c@{\ }c@{\ }c@{\ }c@{\ }c@{\ }c@{\ }c@{\ }c@{\ }c@{\ }c@{\ }c@{\ }c@{\ }c@{\ }c@{\ }c@{\ }c}
1&1&1&1&1&1&1&1&1&1&1&1&1&1&1&1&0&0&0&0&0  \\
\xi^2&\xi^2&\xi^2&\xi^2&
\xi&\xi&\xi&\xi&
1&1&1&1&0&0&0&0 &1&1&1&1&0 \\
\xi^2&\xi&1&0&
\xi^2&\xi&1&0&
\xi^2&\xi&1&0&
\xi^2&\xi&1&0&
\xi^2&\xi&1&0&1
\end{array}\!\!\right].
$$

\section{Construction}
The approach of the construction for 
$1$-perfect codes in $D(m,n)$ is partially similar to that of \cite{KoolMun:2000} for tight $2$-designs (the codes formally dual to $1$-perfect). We start with the Hamming code $\mathcal H$ over $\FF_4$ in $H(2m+n,4)$ and replace subwords of length $4$  corresponding to the positions 
$4i-3$, $4i-2$, $4i-1$, $4i$ of the codewords by 
subwords of length $4$ over ${\mathbb Z}_4$, treated 
as elements of $D(2,0)$ if $i\le [m/2]$
or $D(1,2)$ if $i=(m+1)/2$.

In details, there are some differences with the construction in \cite{KoolMun:2000}.
For the code dual to $\mathcal H$,
there are only $16$ possibilities for
subwords in the considered quadruples of coordinates, and the substitution function used in \cite{KoolMun:2000} is an isometry  from the corresponding subcode in $H(4,4)$ into 
$D(2,0)$ ($D(1,2)$). 
In our case, all $256$ possible length-$4$ words occur as subwords, and there is no such isometry (indeed, the graphs $H(4,4)$,
  $D(1,2)$, $D(2,0)$ are not isomorphic). However, for the resulting code being $1$-perfect, we need not control the distance between any two codewords;
it is sufficient only to ensure that this distance cannot be $1$ or $2$. 
To do this, we construct the substitution bijection between $H(4,4)$
and $D(2,0)$ ($D(1,2)$) using the principles of the generalized concatenated construction \cite{Zin1976:GCC}.
It occurs that the resulting construction is close to a variant of the generalized concatenated construction for $1$-perfect codes in $H(n,q)$ presented in
\cite{Romanov:concat}.

\subsection{Codes in $D(1,2)$, $D(2,0)$ and $H(4,4)$.}

To construct a substitution function with the desired properties, in each of the graphs $D(1,2)$, $D(2,0)$, $H(4,4)$, we need two additive codes, of distance $3$ and $2$ and cardinality $16$ and $64$, respectively.

\begin{lemma}\label{lemma:d4}
Denote
\begin{IEEEeqnarray*}{rCl}
 \bar x = (0,1,2,3),\ 
\bar y = (1,0,1,2) &\in& {\mathbb Z}_4^4;
\\
\bar z = (0,0,1,1) &\in& {\mathbb Z}_4^4;
\\
\bar u = (0,0,0,2),\ \bar v = (0,0,2,0) &\in& {\mathbb Z}_4^4;
\\
\bar x'=(1,1,1,1),\ \bar y'=(0,1,\xi,\xi^2) &\in& \FF_4^4;
\\
\bar z'= (0,0,1,1) &\in& \FF_4^4.
\end{IEEEeqnarray*}
Define \vspace{-2ex}
\begin{IEEEeqnarray*}{rClrCl}
C''&=&
\langle \bar x, \bar y \rangle,
\quad &
C'&=&\langle \bar x, \bar y, \bar z \rangle;
\\
D'' &=&
\langle \bar x, \bar y \rangle,
\quad&
D'&=&\langle \bar x, \bar y, \bar u, \bar v \rangle;
\\
E''&=&
\langle \bar x', \bar y' \rangle,
\quad &
E'&=&
\langle \bar x', \bar y', \bar z' \rangle.
\end{IEEEeqnarray*}
We state that
\begin{itemize}
    \item[(a)] $C''\subset C'$,  
$D''\subset D'$,  
$E''\subset E'$;
    \item[(b)] $C'$, $D'$, $E'$
are distance-$2$ codes 
of cardinality $64$ in 
$D(1,2)$, $D(2,0)$, $H(4,4)$,
respectively;
\item[(c)] $C''$, $D''$, $E''$
are distance-$3$ codes 
of cardinality $16$ in 
$D(1,2)$, $D(2,0)$, $H(4,4)$,
respectively.
\end{itemize}
\end{lemma}
\begin{proof}
We note that since the considered codes are closed under addition,
the code distance coincides with the minimum nonzero weight of a codeword.

(a) is trivial.

(b). Every codeword of $C'$ is orthogonal to $(1,1,1,3)$, 
as this is true for $\bar x$, $\bar y$, and $\bar z$.
It is easy to check that each of the $12$ words 
$(0,0,0,1)$, $(0,0,0,2)$, $(0,0,0,3)$, 
$(0,0,1,0)$, $(0,0,2,0)$, $(0,0,3,0)$, 
$(0,1,0,0)$, $(0,3,0,0)$, 
$(1,0,0,0)$, $(3,0,0,0)$, 
$(1,1,0,0)$, $(3,3,0,0)$
of weight $1$ in $D(1,2)$ is not orthogonal to $(1,1,1,3)$.
Hence, $C'$ does not contain weight-$1$ words and has code distance larger than $1$.
The cardinality of $C'$ is $4\cdot 4\cdot 4$, as $\bar x$, $\bar y$, $\bar z$ are linearly independent.
 
Each  codeword of $D'$ is orthogonal to both $(0,2,0,2)$ and $(2,0,2,0)$,
while this is not true for each 
of the $12$ weight-$1$ words 
$(0,0,0,1)$, $(0,0,0,3)$, 
$(0,0,1,0)$, $(0,0,3,0)$, 
$(0,0,1,1)$, $(0,0,3,3)$,
$(0,1,0,0)$, $(0,3,0,0)$, 
$(1,0,0,0)$, $(3,0,0,0)$, 
$(1,1,0,0)$, $(3,3,0,0)$
in $D(2,0)$.
Hence, $D'$ does not contain weight-$1$ words 
and has code distance larger than $1$.
Since $D'$ is spanned by independent elements 
of order $4$, $4$, $2$, and $2$, 
its cardinality is $4\cdot 4\cdot 2\cdot 2$.

Similar arguments work for $E'$, orthogonal to $(1,1,1,1)$.

(c). Each of $C''$, $D''$, $E''$ is the span of two linearly independent words of order $4$, so the cardinality is $16$ in each case.
Next, it is easy to see that each of the $15$ nontrivial linear combinations of $\bar x'$ and $\bar y'$ has at most one zero symbol; 
so, the minimum weight (and hence the code distance) of $E''$ is $3$.
For the codes $C''$ and $D''$,
the minimum weight (in $D(1,2)$ and $D(2,0)$, respectively)
can be easily found 
from the complete list of codewords:
$$ 
\begin{array}{r@{\ }c@{\ }c@{\ }l}
C''=D''= 
\{(0,0,0,0),&(0,1,2,3),&(0,2,0,2),& (0,3,2,1),\\
 (1,0,1,2),& (1,1,3,1),& (1,2,1,0),& (1,3,3,3),\\
(2,0,2,0),& (2,1,0,3),& (2,2,2,2),&  (2,3,0,1),\\
 (3,0,3,2),& (3,1,1,1),& (3,2,3,0),& (3,3,1,3)\}.
\end{array}
$$

\end{proof}

\begin{lemma}\label{l:e}
Let $\bar c=(c_1,\ldots,c_n)$ be a codeword of the Hamming code $\mathcal H$, and let $\bar e = (e_1,e_2,e_3,e_4)$ be a codeword of the code $E''$ defined in Lemma~\ref{lemma:d4}. Then for every $j$, $0\le j< (n-1)/4$, the word
$\bar b=(b_1,\ldots,b_n)$ whose components  are
$$ b_i=
\left\{ 
\begin{array}{ll}
c_i+e_{i-4j} 
        & \mbox{if}\ i\in\{4j+1,4j+2,4j+3,4j+4\}, \\
        c_i & \mbox{otherwise}
\end{array}
\right.
$$
is also a codeword of $\mathcal H$.
\end{lemma}
\begin{proof}
It is sufficient to prove the statement for the case when $\bar c$ is the all-zero word.


For the all-zero $\bar c$,
the word $\bar b$ has the form
$(0,\ldots,0,e_1,e_2,e_3,e_4,0,\ldots,0)$,
and its syndrome $P \bar b$ coincides 
with $P_{(4j+1,4j+2,4j+3,4j+4)} \bar e$, where the matrix $P_{(4j+1,4j+2,4j+3,4j+4)}$ is composed from the four corresponding columns of $P$. By the construction of $P$ (recall, it consists of all different columns  whose  first  nonzero  element is $1$ placed in the inverse lexicographical order), the considered submatrix has the last row $(\xi^2,\xi,1,0)$, while the other rows are multiples of $(1,1,1,1)$.
From the definition of the code $E''$ in Lemma~\ref{lemma:d4}, we see that its 
codewords are orthogonal to both $(\xi^2,\xi,1,0)$ as $(1,1,1,1)$ (indeed, this is true for the base codewords $\bar x'$ and $\bar y'$). It follows that 
$P_{(4j+1,4j+2,4j+3,4j+4)} \bar e = \bar 0$ and, hence, $P \bar b=\bar 0$. That is, $\bar b$ belongs to $\mathcal H$.
\end{proof}
\begin{lemma}\label{l:cde}
For every two cosets $C''_1$, $C''_2$ 
of $C''$ that are not subsets of the same coset of $C'$, for every $\bar x$ from $C''_1$, there is $\bar y$ from $C''_2$ at distance $1$ from $\bar x$. The same holds for the cosets of $D''$
that are not in one coset of $D'$
and for the cosets of $E''$ 
that are not in one coset of $E'$.
\end{lemma}
\begin{proof}
The statement is proven by the following counting argument. The word $\bar x$ has exactly $12$ neighbors. Two neighbors cannot belong to the same coset of $C''$, because $C''$ is distance-$3$. No one of these $12$ neighbors belongs to the same coset of $C'$ as $\bar x$, because $C'$ is distance-$2$. Since there are $16$ cosets of $C''$ and $4$ of them are subsets of the same coset of $C'$ containing $\bar x$, each of the remaining $12$ cosets contains exactly one neighbor of $\bar x$.
\end{proof}


\subsection{Main theorem}\label{ss:main}

For the construction,
we need two maps, $\phi$ and $\psi$.
The bijective map $\phi$ between the vertex sets 
of $H(4,4)$ and $D(2,0)$
is required to satisfy
the following conditions:
\begin{enumerate}
\item[\rm (a)] $\bar a$ and $\bar b$ 
belong to the same coset of $E''$ if and only if $\phi(\bar a)$ and $\phi(\bar b)$ belong to the same coset of~$D''$;
\item[\rm (b)] $\bar a$ and $\bar b$ belong to the same coset of $E'$ if and only if $\phi(\bar a)$ and $\phi(\bar b)$
belong to the same coset of~$D'$.
\end{enumerate}
To construct $\phi$, we represent each of the $256$
vertices of $H(4,4)$ as 
$\bar e'_i+\bar e''_j +\bar e'''_k$, 
$i,j\in\{0,1,2,3\}$, $k\in\{0,\ldots,15\}$, where
\begin{itemize}
 \item $\bar e'_{0}$, $\bar e'_{1}$, $\bar e'_{2}$, $\bar e'_{3}$ are representatives of the four cosets of~$E'$;
 \item  $\bar e''_{0}$, $\bar e''_{1}$, $\bar e''_{2}$, $\bar e''_{3}$ are representatives of the four cosets of~$E''$ in~$E'$;
 \item $E''=\{\bar e'''_0$, \ldots, $\bar e'''_{15}\}$.
\end{itemize}
Now, the bijection $\phi$ is defined by
$$ \phi(\bar e'_i+\bar e''_j +\bar e'''_k)=
\bar d'_i+\bar d''_j +\bar d'''_k$$
for every $i$ from $\{0,1,2,3\}$,
every $j$ from $\{0,1,2,3\}$,
and every $k$ from $\{0,\ldots,15\}$. 
In a similar manner, each of the $256$
vertices of $D(2,0)$ is represented as 
$\bar d'_i+\bar d''_j +\bar d'''_k$, 
$i,j\in\{0,1,2,3\}$, $k\in\{0,\ldots,15\}$, where
\begin{itemize}
 \item $\bar d'_{0}$, $\bar d'_{1}$, $\bar d'_{2}$, $\bar d'_{3}$ are representatives of the four cosets of~$D'$;
 \item  $\bar d''_{0}$, $\bar d''_{1}$, $\bar d''_{2}$, $\bar d''_{3}$ are representatives of the four cosets of~$D''$ in~$D'$;
 \item $D''=\{\bar d'''_0$, \ldots, $\bar d'''_{15}\}$.
\end{itemize}

The bijective map $\psi$ between the vertex sets 
of $H(4,4)$ and $D(1,2)$ is constructed similarly,
involving the cosets of $C'$ and $C''$,
and satisfies the following conditions:
\begin{enumerate}
    \item[\rm (c)]$\bar a$ and $\bar b$ 
belong to the same coset of~$E''$
if and only if 
$\psi(\bar a)$ and $\psi(\bar b)$
belong to the same coset~of~$C''$;
    \item[\rm (d)] $\bar a$ and $\bar b$ 
belong to the same coset of~$E'$
if and only if 
$\psi(\bar a)$ and $\psi(\bar b)$
belong to the same coset of~$C'$.
\end{enumerate}

\begin{theorem}\label{th:main}
Let $\mathcal H$ be the Hamming code in $H((4^k-1)/3,4)$ whose check matrix consists 
of all columns with first nonzero element $1$, in the inverse lexicographical order.
Let the codes $E''$, $E'$ 
in $H(4,4)$,
the codes $C''$, $C'$ 
in $D(1,2)$,
the codes $D''$, $D'$ 
in $D(2,0)$ be defined as in Lemma~\ref{lemma:d4}.
Let $\phi$ and $\psi$ be bijective maps 
from the vertex set of $H(4,4)$ to the vertex sets
of $D(2,0)$ and $D(1,2)$, respectively, satisfying
conditions (a), (b) and (c), (d) above.
Let $m$ and $n$ be positive integers 
such that $2m+n=(4^k-1)/3$.
If $m$ is even, then 
\begin{multline*}
\mathcal C = \Big\{ \big(
\vphantom{|^2}
\phi(x_1,...,x_4), \ldots,
\phi(x_{2m-3},...,x_{2m}),\\
\ \ \ \ \ \ \ \ \ \ \ \ \ \ \ \ \ \ \ \ 
x_{2m+1}, \ldots,x_{2m+n} \big):\ \\
(x_{1}, \ldots,x_{2m+n})
\in \mathcal{H}\Big\} 
\end{multline*}
is a $1$-perfect code in $D(m,n)$.
If $m$ is odd, then 
\begin{multline*}
\mathcal C = \Big\{ \big(\phi(x_1,...,x_4), \ldots,
\phi(x_{2m{-}5},...,x_{2m{-}2}), \\
\ \ \ \ 
\psi(x_{2m{-}1},...,x_{2m{+}2}), 
x_{2m+1}, \ldots,x_{2m+n} \big):\ \\
(x_{1}, \ldots,x_{2m+n})
\in \mathcal{H}\Big\} 
\end{multline*}
is a $1$-perfect code in $D(m,n)$.
\end{theorem}

\begin{proof}
We will consider the case when $m$ is even; the odd case is similar.
Assume the receiver get a word 
$\bar y =(y_1,\ldots,y_{2m+n}) \in {\mathbb Z}_4^{2m} \times \FF_4^n $, associated
with a vertex of $D(m,n)$.
To decode the message under the assumption that 
an error of weight at most $1$ occurred,
one should find a codeword $\bar c$ 
at distance at most $1$ from $\bar y$.
Consider 
\begin{IEEEeqnarray*}{rl}
\bar x = \big(\phi^{-1} (y_1,y_2), \ldots, \phi^{-1} 
(y_{2m-1},y_{2m}),  y_{2m+1}, \ldots, y_{2m+n}
 \big) \\ \in {\FF_4}^{2m+n}. 
\end{IEEEeqnarray*}
If $\bar x$ is a codeword of $\mathcal H$, then, by the definition of $\mathcal C$, we have $\bar c=\bar y\in \mathcal C$.
Assume that $\bar x\not\in\mathcal H$. Since $\mathcal H$ is a $1$-perfect code, there is $\bar b=(b_1,\ldots,b_{2m+n})\in\mathcal H$ at distance $1$ from $\bar x$.
We consider the codeword $\bar z\in \mathcal C$ defined as 
\begin{IEEEeqnarray*}{l}
\bar z = (z_1,\ldots,z_{2m+n}) 
\\
= \big(\phi(b_1,b_2,b_3,b_4), \ldots,  
\phi(b_{2m-3},
...,
b_{2m}), 
 b_{2m+1},\ldots,b_{2m+n}\big).
\end{IEEEeqnarray*}
Note that $\bar z$ is not necessarily the required $\bar c$.
However, we can state the following.
\begin{itemize}
 \item[(i)] \emph{If $\bar b$ differs from $ \bar x$ in one of the last $n$ coordinates, then $ \bar z$ and $ \bar y$ differ in exactly one, the same as $\bar b$ and $\bar x$,  coordinate; so, $\bar c= \bar z$ in this case.} 
 Indeed, $\bar z$ and $\bar y$ trivially coincide in the other coordinates.
 \item[(ii)] \emph{If $\bar b$ differs from $\bar x$ in one of the first $2m$ coordinates,
 say,  
 $(b_{4i-3},b_{4i-2},b_{4i-1},b_{4i})\ne (x_{4i-3},x_{4i-2},x_{4i-1},x_{4i})$, then there is $(c_{4i-3},c_{4i-2},c_{4i-1},c_{4i}) \in {\mathbb Z}_4^4$ in the same coset of $D''$ as $(z_{4i-3},z_{4i-2},z_{4i-1},z_{4i})$ such that
 \begin{multline*} \bar c=(z_1,\ldots,z_{4i-4}, c_{4i-3},c_{4i-2},c_{4i-1},c_{4i},\\
 z_{4i+1},\ldots,z_{2m+n})
 \end{multline*}
 at distance $1$ from $\bar y$. 
 Moreover, $ \bar c$ is a codeword of $\mathcal C$.
 } Indeed, the first part of the claim is straightforward from Lemma~\ref{l:cde} and the definition of the map $\phi$. From Lemma~\ref{l:e} and the construction of $\mathcal C$, we have $ \bar c \in \mathcal C$.
\end{itemize}

 In any case, there is a codeword 
 $\bar c \in \mathcal C$ 
 at distance at most $1$ from $\bar y$. From standard counting arguments (the size of the space equals the size of the code multiplied by the size of a radius-$1$ ball), we see that such a codeword is unique. Therefore, the code is $1$-perfect.
\end{proof}

So, if there is a $1$-perfect code in a $4$-ary Hamming graph, then there is a $1$-perfect code in every Doob graph of the same diameter.

\begin{corollary}
The Doob graph $D(m, n)$ has a non-trivial $e$-perfect code if and only if $e = 1$ and
there is a positive integer $k$ such that 
$2m + n = (4^k-1)/3$.
\end{corollary}
\begin{proof}
The ``if'' and ``only if'' parts of the statement come from Theorem~\ref{th:main}
and \cite[Theorem~3]{KoolMun:2000}, respectively.
\end{proof}

\section{Concluding remarks}

In this section we briefly discuss some related questions,
including those suggested by the reviewers. 

\begin{remark}[Unrestricted $1$-perfect codes vs additive  $1$-perfect codes]
For every Doob graph $D(m,n)$ that satisfies 
the obvious ball-packing necessary condition 
on the existence of $1$-perfect codes, 
we can construct such a code 
by Theorem~\ref{th:main}. 
In general, the code constructed 
is not linear or even additive 
(closed with respect to addition). 
Moreover, as was shown in
\cite[Theorem~1]{Kro:perfect-doob}, 
the existence of additive $1$-perfect codes 
implies additional conditions 
on the parameters $m$ and $n$. 
Namely, \vspace{-5mm}
\begin{IEEEeqnarray}{rCl} \nonumber
2m+n&=&(2^{\Gamma+2\Delta}-1)/3,\\ \label{eq:additive}
3n&=&2^{\Gamma+\Delta}-1-2n'', \\ \nonumber
1 &\ne& n''\le 2^\Delta -1
\end{IEEEeqnarray}
for some nonnegative integer 
$\Gamma$, $\Delta$, $n''$. 
Examples of Doob graphs 
for which additive $1$-perfect codes do not exist,
while unrestricted $1$-perfect codes
can be constructed by Theorem~\ref{th:main}, are
$D(6,9)$, $D(9,3)$, $D(10,1)$.
As can be seen from the proof of the theorem, 
we do not need additivity 
to have a good decoding algorithm. 
Indeed, decoding the constructed code 
in the Doob graph 
is not more complicate 
than decoding the original 
$4$-ary Hamming code of length $2m+n$; 
all additional operations 
(mainly, applying $\phi$ and $\phi^{-1}$)
take $o(2m+n)$ time.
\end{remark}

\begin{remark}[dual codes]
The codes formally dual to the $1$-perfect codes are known as {simplex codes}
or tight $2$-designs (the formal duality of two codes means that the MacWilliams transform of the distance distribution of one code gives the distance distribution of the other code).
A \emph{simplex code} in a Hamming graph $H(N,q)$
or a Doob graph $D(m,n)$ has 
$N(q-1)+1$ codewords at mutual distance 
$(N(q-1)+1)/q$ from each other
(for $D(m,n)$, we put $N=2m+n$ and $q=4$).
In every $D(m,n)$ such that $2m+n=(4^k-1)/3$ for some $k$, 
simplex codes
were constructed in \cite{KoolMun:2000}. 
So, it is safe to say that for every $1$-perfect code $C$ in a Doob graph
there is a simplex code (tight $2$-design) that is formally dual to $C$.
However, to treat duality in the usual sense, 
as a duality between two submodules,
we need additive codes 
(note that the duality should be defined in a special way to be compatible with the MacWilliams transform, see~\cite{Krotov:ISIT2019:Doob}).
Since additive $1$-perfect codes (and hence, additive simplex codes)
exist only if the parameters satisfy the additional conditions \eqref{eq:additive},
for any other parameters meeting $2m+n=(4^k-1)/3$ 
the class of $1$-perfect codes is connected with the class of simplex codes
only in the way of formal duality (there is still a challenge in finding 
a more strict connection, as was done, for example, in \cite{HammonsOth:Z4_linearity} for the formally dual 
classes of Preparata-like codes and Kerdock-like codes).
It should also be noted that the parameters of simplex codes (tight $2$-designs)
are not proven to be restricted by the case $2m+n=(4^k-1)/3$ 
only. The problem of existence of simplex codes for other parameters 
is open for Doob graphs, as well as for Hamming graphs,
where the most attempts are focused on the binary case (see the Hadamard conjecture).
\end{remark}

\begin{remark}[from $H(2m+n,4)$ to  $D(m,n)$]
To solve the problem of parameters of perfect codes in Doob graphs, 
we apply the strategy of ``switching'' between
the graphs $D(m,n)$ and $H(2m+n,4)$
by transforming the Shrikhande components of the Cartesian product, in groups of one or two, 
to Hamming ($K_4$) components.
This general approach is not new and 
was applied for different purposes before;
however, the realizations depend on concrete problems.
In \cite{KoolMun:2000}, 
isometries between special vertex sets in 
$H(4,4)$, $D(1,2)$, and $D(2,0)$ were used to construct 
a simplex code in a Doob graph from a special (not arbitrary) simplex code in a Hamming graph. 
In \cite{Kro:2015:N-MDS-Doob}, any maximum independent set in a Doob graph is mapped to a maximum independent set (unrestricted distance-$2$ MDS code) in the corresponding Hamming graph;
again, the map is constructed based 
on a map in the smallest case, from $D(1,0)$ to $H(2,4)$,
but the map is set-to-set and cannot be treated in 
a point-wise manner. 
In the current paper, we use point-to-point maps $\phi$ and $\psi$ that preserve
some metrical properties 
of a special coset partition in $H(4,4)$
to construct perfect codes in Doob graphs 
from a special Hamming code 
(we cannot apply the same maps 
to an arbitrary $1$-perfect code
or even to an equivalent Hamming code 
because we need to control 
which subcodes occur 
in subsequent groups of coordinates).
As shown in \cite{Kubota:2017}, 
one can obtain $D(m,n)$ from $H(2m+n,4)$ 
by a sequence of $m$ Godsil--McKay switchings,
each switching replacing one Shrikhande component 
in the Cartesian product by two $K_4$-components.
Godsil--McKay switching can be treated as a bigective map
between the vertex set of two graphs, and it is also 
induces, in an algebraic way, an isomorphism between
eigenspaces of the graphs with the same eigenvalues
(so the graphs related by Godsil--McKay switching are always cospectral). However, it obviously changes the distances 
between some vertices, 
and so cannot serve the purpose of constructing error-correcting codes
in a straightforward way.
It is still a very interesting question
if some bijections $\phi$ and $\psi$ with the desired 
properties (Section~\ref{ss:main}) can be treated as 
a combination of Godsil--McKay switchings, but even if they can,
this would hardly simplify the construction or its proof.
\end{remark}

%

\renewcommand\refname{}

\vspace{-0.99cm}
 
 \providecommand\href[2]{#2} \providecommand\url[1]{\href{#1}{#1}}
  \def\DOI#1{{\small {DOI}:
  \href{http://dx.doi.org/#1}{#1}}}\def\DOIURL#1#2{{\small{DOI}:
  \href{http://dx.doi.org/#2}{#1}}}

 \end{document}